\documentclass[11pt,reqno]{amsart}
\usepackage[cm]{fullpage}
\usepackage{mathrsfs,amssymb,graphicx,verbatim,amsmath,amsfonts}
\usepackage{paralist}
\usepackage[breaklinks,pdfstartview=FitH]{hyperref}
\usepackage{upgreek}
\usepackage{tikz}
\usepackage[mathscr]{euscript}
\usepackage[T1]{fontenc}
\usepackage{dutchcal}
\usepackage{fourier}
\usepackage{multirow}
\usepackage{caption}

\renewcommand{\le}{\leqslant}
\renewcommand{\ge}{\geqslant}

\renewcommand{\setminus}{\smallsetminus}
\allowdisplaybreaks
\newcommand{\ud}[0]{\,\mathrm{d}}
\usepackage{esint}
\usepackage[autostyle]{csquotes}

\newcommand{\n}{\{1,\ldots,n\}}

\renewcommand{\k}{\{1,\ldots,k\}}

\newcommand{\p}{\mathfrak{p}}
\newcommand{\q}{\mathfrak{q}}
\renewcommand{\t}{\vartheta}

\renewcommand{\d}{\delta}

\newcommand{\e}{\varepsilon}
\newcommand{\R}{\mathbb R}
\newcommand{\1}{\mathbf 1}

\newtheorem{theorem}{Theorem}
\newtheorem{lemma}[theorem]{Lemma}

\newtheorem{fact}[theorem]{Fact}

\theoremstyle{remark}
\newtheorem{remark}[theorem]{Remark}

\newtheorem{question}[theorem]{Question}

\renewcommand{\O}{\mathscr{O}}

\renewcommand{\subset}{\subseteq}

\newcommand{\W}{\mathsf{W}}

\newcommand{\f}{\varphi}

\newcommand{\F}{\mathbb F}
\newcommand{\N}{\mathbb N}

\newcommand{\eqdef}{\stackrel{\mathrm{def}}{=}}

\newcommand{\cc}{\mathsf{c}}

\renewcommand{\emptyset}{\varnothing}

\begin{document}

\title{The Andoni--Krauthgamer--Razenshteyn characterization of sketchable norms \\ fails for sketchable metrics}

\author{Subhash Khot}

\address{(S.K.) Department of Computer Science, Courant Institute of Mathematical Sciences\\ New York University\\251 Mercer Street, New York, NY 10012-1185, USA}
\email{khot@cs.nyu.edu}

\author{Assaf Naor}
\address{(A.N.) Mathematics Department\\ Princeton University\\ Fine Hall, Washington Road, Princeton, NJ 08544-1000, USA}
\email{naor@math.princeton.edu}

\thanks{S.K. was supported by NSF CCF-1422159 and the Simons Foundation. A.N. was supported by NSF CCF-1412958, the Packard Foundation and the Simons Foundation. This work  was carried out under the auspices of the Simons Algorithms and Geometry (A\&G) Think Tank.}


\maketitle


\begin{abstract}
Andoni, Krauthgamer and Razenshteyn (AKR)  proved (STOC 2015) that a finite-dimensional normed space $(X,\|\cdot\|_X)$ admits a $O(1)$ sketching algorithm (namely, with $O(1)$ sketch size and $O(1)$ approximation) if and only if for every $\e\in (0,1)$ there exist $\alpha\ge 1$ and an embedding $f:X\to \ell_{1-\e}$ such that $\|x-y\|_X\le \|f(x)-f(y)\|_{1-\e}\le \alpha \|x-y\|_X$ for all $x,y\in X$. The "if part" of this theorem follows from a sketching algorithm of Indyk (FOCS 2000). The contribution of AKR is therefore to demonstrate that the mere availability  of a sketching algorithm implies the existence of the aforementioned geometric realization. Indyk's algorithm shows that the "if part" of the  AKR characterization holds true for any metric space whatsoever, i.e., the existence of an embedding as above implies sketchability even when $X$ is not  a normed space. Due to this, a natural question that AKR posed was whether the assumption that the underlying space is a normed space is needed for  their characterization of sketchability. We resolve this question by proving that for arbitrarily large $n\in \N$ there is an $n$-point metric space $(M(n),d_{M(n)})$ which is $O(1)$-sketchable yet for every $\e\in (0,\frac12)$, if $\alpha(n)\ge 1$ and $f_n:M(n)\to \ell_{1-\e}$ are such that $d_{M(n)}(x,y)\le \|f_n(x)-f_n(y)\|_{1-\e}\le \alpha(n) d_{M(n)}(x,y)$ for all $x,y\in M(n)$, then  necessarily $\lim_{n\to \infty} \alpha(n)= \infty$.
\end{abstract}

\section{Introduction}

We shall start by recalling the notion of sketchability; it is implicit in seminal work~\cite{AMS99} of  Alon, Matias and Szegedy, though the formal definition that is described below was put forth  by   Saks and Sun~\cite{SS02}. This is a crucial and well-studied algorithmic primitive for analyzing massive date sets, with several powerful applications; surveying them here would be needlessly repetitive, so we refer instead to e.g.~\cite{Ind06,AKR18} and the references therein.

Given a set $X$, a function  $K:X\times X\to \R$  is called a nonnegative kernel if $K(x,y)\ge 0$ and $K(x,y)=K(y,x)$ for every $x,y\in X$. In what follows, we will be mainly interested in the geometric setting when the kernel $K=d_X$ is in fact a metric on $X$, but even for that purpose we will also need to consider nonnegative kernels that are not metrics.

Fix $D\ge 1$ and $s\in \N$. Say that a nonnegative kernel $K:X\times X\to [0,\infty)$ is $(s,D)$-sketchable if for every $r>0$ there is a mapping $\mathsf{R}=\mathsf{R}_r:\{0,1\}^s\times \{0,1\}^s\to \{0,1\}$ and a probability distribution over  mappings $\mathsf{Sk}=\mathsf{Sk}_r:X\to \{0,1\}^s$  such that
\begin{equation}\label{def:sketchable}
\inf_{\substack{x,y\in X\\ K(x,y)\le r}} \mathrm{\bf Prob} \Big[\mathsf{R}\big(\mathsf{Sk}(x),\mathsf{Sk}(y)\big)=0\Big]\ge \frac35\qquad\mathrm{and}\qquad \inf_{\substack{x,y\in X\\ K(x,y)>D r}} \mathrm{\bf Prob} \Big[\mathsf{R}\big(\mathsf{Sk}(x),\mathsf{Sk}(y)\big)=1\Big]\ge \frac35.
\end{equation}
The value $\frac35$ in~\eqref{def:sketchable} can be replaced throughout by any constant that is strictly bigger than $\frac12$; we chose to fix an arbitrary value here in order to avoid the need for the notation to indicate dependence on a further parameter. A kernel (or, more formally, a family of kernels) is said to be sketchable if it is $(s,D)$-sketchable  for some $s=O(1)$ and $D=O(1)$.


The way to interpret the above definition is to think of $\mathsf{Sk}$ as a randomized method to assign one of the $2^s$ labels $\{0,1\}^s$ to each point in $X$, and to think of $\mathsf{R}$ as a reconstruction algorithm that takes as input two such labels in $\{0,1\}^s$  and outputs either $0$ or $1$, which stand for  "small" or "large," respectively. The meaning of~\eqref{def:sketchable} becomes  that for every pair $x,y\in X$, if one applies the reconstruction algorithm to the random labels $\mathsf{Sk}(x)$ and $\mathsf{Sk}(y)$, then with substantially high probability  its output is consistent with the value of the kernel $K(x,y)$ at scale $r$ and approximation $D$, namely the algorithm declares "small" if $K(x,y)$ is at most $r$, and it declares "large" if $K(x,y)$ is greater than $D r$.

Suppose that $\alpha,\beta,\theta>0$ and that $K:X\times X\to [0,\infty)$ and $L:Y\times Y\to [0,\infty)$ are nonnegative kernels on the sets $X$ and $Y$, respectively. Suppose also that there is $f:Y\to X$ such that $\alpha L(x,y)^\theta\le K(f(x),f(y))\le \beta L(x,y)^\theta$ for all $x,y\in Y$. It follows formally from this assumption and the above definition that if $K$ is $(s,D)$-sketchable for some $s\in \N$ and $D\ge 1$, then $L$ is $(s,(\beta D/\alpha)^{1/\theta})$-sketchable. Such an "embedding approach" to deduce sketchability is used frequently in the literature. As an example of its many consequences, since $\ell_2$ is sketchable by  the works of Indyk and Motwani~\cite{IM99} and Kushilevitz, Ostrovsky and Rabani~\cite{KOR00}, so is any metric space of negative type, where we recall that a metric space $(X,d)$ is said to be of negative type (see e.g.~\cite{DL97}) if the metric space $(X,\rho)$ with $\rho=\sqrt{d}$ is isometric to a subset of $\ell_2$.


\subsection{The Andoni--Krauthgamer--Razenshteyn characterization of sketchable norms} The following theorem from~\cite{AKR18} is a remarkable result of Andoni, Krauthgamer and Razenshteyn (AKR) that characterizes  those norms that are sketchable\footnote{In~\cite{AKR18}, the conclusion of Theorem~\ref{thm:AKR quote} is proven under a formally weaker assumption, namely  it uses a less stringent notion of sketchability  which allows for the random sketches of the points $x,y\in X$ to be different from each other, and for the reconstruction algorithm to depend on the underlying randomness that was used to produce those sketches. Since our main result, namely Theorem~\ref{thm:our sketching statement} below, is an impossibility statement, it becomes only stronger if we use the simpler and stronger notion of sketchability that we stated above.} in terms of their geometric embeddability into a classical kernel (which is not a metric).

\begin{theorem}[AKR characterization of sketchability]\label{thm:AKR quote} Fix $s\in \N$ and $D\ge 1$. A finite-dimensional normed space $(X,\|\cdot\|_X)$ is $(s,D)$-sketchable if and only if for any $\e\in (0,1)$ there exists $\alpha=\alpha(s,D,\e)>0$ and an embedding $f:X\to \ell_{1-\e}$ such that
\begin{equation*}\label{eq:l 1-eps conclusion}
\forall\, x,y\in X,\qquad \|x-y\|_X\le  \|f(x)-f(y)\|_{1-\e}\le \alpha\|x-y\|_X.
\end{equation*}
\end{theorem}
Thus, a finite-dimensional normed space  is sketchable if and only if it can be realized as a subset of a the classical sequence space $\ell_{1-\e}$ so that the kernel $\|\cdot\|_{1-\e}$ reproduces faithfully (namely, up to factor $\alpha$) all the pairwise distances in $X$. See~\cite[Theorem~1.2]{AKR18} for an explicit dependence in Theorem~\ref{thm:AKR quote} of $\alpha(s,D,\e)$ on the parameters $s,D,\e$.

\subsubsection*{$L_p$ space notation}In Theorem~\ref{thm:AKR quote} and below, we use the following standard notation for $L_p$ spaces. If $p\in (0,\infty)$ and $(\Omega,\mu)$ is a measure space, then $L_p(\mu)$ is the set of (equivalence classes up to measure $0$ of) measurable functions $\f:\Omega\to \R$ with $\int_\Omega |\f(\omega)|^p\ud\mu(\omega)<\infty$. When $\mu$ is the counting measure on $\N$, write $L_p(\mu)=\ell_p$. When $\mu$ is the counting measure on $\n$ for some $n\in \N$, write $L_p(\mu)=\ell_p^n$. When $\mu$ is the Lebesgue measure on $[0,1]$, write $L_p(\mu)=L_p$. When the underlying measure is clear from the context (e.g.~counting measure or Lebesgue measure), one sometimes writes $L_p(\mu)=L_p(\Omega)$. The $L_p(\mu)$ (quasi)norm is defined by setting $\|\f\|_p^p=\int_\Omega |\f(\omega)|^p\ud\mu(\omega)$ for $\f\in L_p(\mu)$.
While if $p\ge 1$, then $(\f,\psi)\mapsto \|\f-\psi\|_p$  is a metric on $L_p(\mu)$, if $p=1-\e$ for some $\e\in (0,1)$, then $\|\cdot\|_{1-\e}$ is not a metric; if $L_{1-\e}(\mu)$ is infinite dimensional, then $\|\cdot\|_{1-\e}$  is not even equivalent to a metric in the sense that there do not exist any $c,C\in (0,\infty)$ and a metric $d:L_{1-\e}(\mu)\times L_{1-\e}(\mu)\to [0,\infty)$ such that $cd(\f,\psi)\le \|\f-\psi\|_{1-\e}\le Cd(\f,\psi)$ for all $\f,\psi\in L_p(\mu)$. Nevertheless, $\|\cdot\|_{1-\e}$ is  a nonnegative kernel on $L_p(\mu)$ and there is a canonical metric $\mathfrak{d}_{1-\e}$ on $L_p(\mu)$, which is given by
\begin{equation}\label{eq:metric on <1}
\forall\, \f,\psi\in L_{1-\e}(\mu),\qquad \mathfrak{d}_{1-\e}(\f,\psi)\eqdef \|\f-\psi\|_{1-\e}^{1-\e}=\int_{\Omega}|\f(\omega)-\psi(\omega)|^{1-\e}\ud \mu(\omega).
\end{equation}
See the books~\cite{LT77,LT79} and~\cite{KPR84} for much more on the structure for $L_p(\mu)$ spaces when $p\ge 1$ and $0<p< 1$, respectively.

\subsubsection{Beyond norms?} Fix $\e\in (0,1)$. The sketchability of  the nonnegative kernel on $\ell_{1-\e}$ that is given by $\|\f-\psi\|_{1-\e}$ for $\f,\psi\in \ell_{1-\e}$  was proved by Indyk~\cite{Ind06} (formally, using the above  terminology it is sketchable provided $\e$ is bounded away from $0$; when $\e\to 0$ the space $s=s(\e)$ of Indyk's algorithm becomes unbounded). Thus, any metric space $(M,d_M)$ for which there exists $\alpha\in [1,\infty)$ and  an embedding $f:M\to \ell_{1-\e}$ that satisfies
\begin{equation}\label{eq:metric version into 1-eps}
\forall\, x,y\in M,\qquad d_M(x,y)\le \|f(x)-f(y)\|_{1-\e}\le \alpha d_M(x,y)
\end{equation}
is sketchable with sketch size $O_\e(1)$ and approximation $O(\alpha)$. Therefore, the "if part" of Theorem~\ref{thm:AKR quote}  holds for any metric space whatsoever, not only for norms. The "only if" part of Theorem~\ref{thm:AKR quote}, namely showing that the mere  availability of a sketching algorithm for a normed space implies that it can be realized faithfully as a subset of $\ell_{1-\e}$, is the main result of~\cite{AKR18}. This major achievement demonstrates that a fundamental algorithmic primitive {\em coincides} with a  geometric/analytic property that has been studied long before sketchability was introduced (other phenomena of this nature were discovered in the literature, but they are rare). The underlying reason for Theorem~\ref{thm:AKR quote} is deep, as the proof in~\cite{AKR18} relies on a combination of major results from the literature on functional analysis and communication complexity.

A natural question that Theorem~\ref{thm:AKR quote} leaves open is whether one could obtain the same result for $\e=0$, namely for embeddings into $\ell_1$. As discussed in~\cite{AKR18}, this is equivalent to an old question~\cite{Kwa70} of Kwapie\'n; a positive result in this direction (for a certain class of norms) is derived in~\cite{AKR18} using classical partial progress of Kalton~\cite{Kal85} on Kwapie\'n's problem, but fully answering this longstanding question  seems  difficult (and it may very well have a negative answer).

Another natural question that Theorem~\ref{thm:AKR quote} leaves open is whether its assumption that the underlying metric space is a norm is needed. Given that the "if part" of Theorem~\ref{thm:AKR quote}  holds for any metric space, this amounts to understanding whether  a sketchable metric space $(M,d_M)$ admits for every $\e\in (0,1)$ an embedding $f:M\to \ell_{1-\e}$ that satisfies~\eqref{eq:metric version into 1-eps}. This was a central open question of~\cite{AKR18}. Theorem~\ref{thm:our sketching statement} below resolves this question. It should be noted that the authors of~\cite{AKR18} formulated their question while hinting that they suspect that the answer is negative, namely in~\cite[page~893]{AKR18} they wrote {\em "we are not aware of any counter-example to the generalization of Theorem~1.2 to general metrics"} (Theorem~1.2 in~\cite{AKR18} corresponds to Theorem~\ref{thm:AKR quote} here). One could therefore view Theorem~\ref{thm:our sketching statement} as a confirmation of a prediction of~\cite{AKR18}.

\begin{theorem}[failure of the AKR characterization for general metrics]\label{thm:our sketching statement} For arbitrarily large $n\in \N$ there exists an $n$-point metric space $(M(n),d_{M(n)})$ which is $\big(O(1),O(1)\big)$-sketchable, yet for every $\e\in \big(0,\frac12\big)$ and $\alpha\ge 1$, if there were a mapping $f:M(n)\to \ell_{1-\e}$ that satisfies $d_{M(n)}(x,y)\le \|f(x)-f(y)\|_{1-\e}\le \alpha d_{M(n)}(x,y)$ for all $x,y\in M(n)$, then necessarily
\begin{equation}\label{eq:our sketching lower}
\alpha\gtrsim (\log\log n)^{\frac{1-2\e}{2(1-\e)}}.
\end{equation}
\end{theorem}

\noindent{\em Asymptotic notation.} In addition to the usual $"O(\cdot),o(\cdot), \Omega(\cdot),\Theta(\cdot)"$  notation, it will be convenient to use throughout this article (as we already did in~\eqref{eq:our sketching lower}) the following (also standard) asymptotic notation. Given two quantities $Q,Q'>0$, the notations
$Q\lesssim Q'$ and $Q'\gtrsim Q$ mean that $Q\le CQ'$ for some
universal constant $C>0$. The notation $Q\asymp Q'$
stands for $(Q\lesssim Q') \wedge  (Q'\lesssim Q)$. If  we need to allow for dependence on parameters, we indicate this by subscripts. For example, in the presence of  auxiliary objects (e.g.~numbers or spaces) $\phi,\mathfrak{Z}$, the notation $Q\lesssim_{\phi,\mathfrak{Z}} Q'$ means that $Q\le C(\phi,\mathfrak{Z})Q' $, where $C(\phi,\mathfrak{Z}) >0$ is allowed to depend only on $\phi,\mathfrak{Z}$; similarly for the notations $Q\gtrsim_{\phi,\mathfrak{Z}} Q'$ and $Q\asymp_{\phi,\mathfrak{Z}} Q'$.

\medskip

We will see that the metric spaces $\{(M(n),d_{M(n)})\}_{n=1}^\infty$ of Theorem~\ref{thm:our sketching statement} are of negative type, so by the above discussion their sketchability follows from the sketchability of Hilbert space~\cite{IM99,KOR00}. In fact, these metric spaces are (subsets of) the metric spaces of negative type that were considered by Devanur, Khot, Saket and Vishnoi in~\cite{DKSV06} as integrality gap examples for the Goemans--Linial semidefinite relaxation of the Sparsest Cut problem with uniform demands. Hence, our contribution  is the geometric aspect of Theorem~\ref{thm:our sketching statement}, namely demonstrating  the non-embeddability  into $\ell_{1-\e}$, rather than its algorithmic component (sketchability). This is a special case of the more general geometric phenomenon of Theorem~\ref{thm:nonembed main} below, which is our main result. It amounts to  strengthening our work~\cite{KN06} which investigated the $\ell_1$ non-embeddability of quotients of metric spaces using Fourier-analytic techniques. Here, we derive the (formally stronger) non-embeddability into $\ell_1$ of snowflakes of such quotients (the relevant terminology is recalled in Section~\ref{sec:embeddings} below). It suffices to mention at this juncture (with further discussion in Section~\ref{Bourgain} below) that on a conceptual level,  the strategy of~\cite{KN06} (as well as that of~\cite{KR09,DKSV06}) for proving non-embeddability using the classical theorem~\cite{KKL88} of Kahn, Kalai and Linial (KKL) on influences of variables does not imply  the required $\ell_1$ non-embeddability of snowflakes of quotients. Instead, we revisit the use of Bourgain's noise sensitivity theorem~\cite{Bou02}, which was applied for other (non-embeddability) purposes in~\cite{KV15,KN06}, but subsequent work~\cite{KR09,DKSV06} realized  that one could use the much simpler KKL theorem in those contexts (even yielding quantitative improvements). Thus, prior to the present work it seemed that, after all, Bourgain's theorem does not have a decisive use in metric embedding theory, but here we see that in fact it has a qualitative advantage over the KKL theorem in some geometric applications.

The present work also shows that the Khot--Vishnoi approach~\cite{KV15} to the  Sparsest Cut integrality gap has a further  qualitative advantage (beyond its relevance to the case of uniform demands) over the use of the Heisenberg group for this purpose~\cite{LN06}, which yields a better~\cite{CKN11} (essentially sharp~\cite{NY18}) lower bound. Indeed, the Heisenberg group is a $O(1)$-doubling metric space (see e.g.~\cite{Hei01}), and by Assouad's embedding theorem~\cite{Ass83} any such  space admits for any $\e\in (0,1)$ an embedding into $\ell_{1-\e}$ which satisfies~\eqref{eq:metric version into 1-eps} with $\alpha\lesssim_\e 1$ (for the connection to Assouad's theorem, which may not be apparent at this point, see Fact~\ref{fact:snowflake} below). Thus, despite its quantitative superiority as an integrality gap example for Sparsest Cut with general demands, the Heisenberg group cannot yield Theorem~\ref{thm:our sketching statement} while the Khot--Vishnoi spaces do (strictly speaking, we work here with a simpler different construction than that of~\cite{KV15}, but an inspection of the ensuing proof reveals that one could have also used the metric spaces of~\cite{KV15} to answer the question of~\cite{AKR18}).

\begin{question} The obvious question that is left open by Theorem~\ref{thm:our sketching statement} is to understand what happens when $\e\in \big[\frac12,1\big)$. While we established a marked qualitative gap vis \`a vis sketchability between the behaviors of general normed spaces and general metric spaces, the possibility remains that there exists some $\e_0\in \big[\frac12,1\big)$ such that any sketchable metric space $(M,d_M)$ admits an embedding into $\ell_{1-\e_0}$ that satisfies~\eqref{eq:metric version into 1-eps} with $\alpha=O(1)$; perhaps one could even take $\e_0=\frac12$ here. This possibility is of course tantalizing, as it would be a complete characterization of sketchable metric spaces  that is nevertheless qualitatively different from its counterpart for general normed spaces. At present, there is insufficient evidence to speculate that this is so, and it seems more likely that other counterexamples could yield a statement that is analogous to Theorem~\ref{thm:our sketching statement} also in the range $\e\in \big[\frac12,1\big)$, though a new idea would be needed for that.
\end{question}

\begin{question} Even in the range $\e\in \big(0,\frac12\big)$ of Theorem~\ref{thm:our sketching statement}, it would be interesting to determine if one could improve~\eqref{eq:our sketching lower} to $\alpha \gtrsim (\log n)^{c(\e)}$ for some $c(\e)>0$ (see Remark~\ref{rem:KM} below for a technical enhancement that yields an asymptotic improvement of~\eqref{eq:our sketching lower}  but does not achieve such a bound). For the corresponding question when $\e=0$, namely embeddings into $\ell_1$, it follows from~\cite{NY18} that one could improve~\eqref{eq:our sketching lower} to $\alpha\gtrsim \sqrt{\log n}$. However, the example that exhibits this stronger lower bound for $\e=0$ is a doubling metric space, and hence by Assouad's theorem~\cite{Ass83} for every $\e>0$ it does admit an embedding into $\ell_{1-\e}$ that satisfies~\eqref{eq:our sketching lower} with $\alpha\lesssim_\e 1$. Note that by~\cite{NRS05,ALN08} we see that if an $n$-point metric space $(M,d_M)$ is sketchable for the reason that for some $\theta\in (0,1]$ the metric space $(M,d_M^\theta)$ is bi-Lipschitz to a subset of $\ell_2$, then~\eqref{eq:our sketching lower} holds for $\e=0$ and $\alpha\lesssim (\log n)^{1/2+o(1)}$. It would be worthwhile to determine if this upper bound on $\alpha$ (for $\e=0$) holds for any sketchable metric space whatsoever, i.e., not only for those whose sketchability is due to the fact that some power of the metric is Hilbertian. It seems plausible that the latter question is accessible using available methods.
\end{question}

\begin{remark}\label{rem:KM} The lower bound~\eqref{eq:our sketching lower} can be improved by incorporating the "enhanced short code argument" of  Kane and Meka~\cite{KM13} (which is in essence a derandomization step) into the ensuing reasoning. This yields a  more complicated construction for which~\eqref{eq:our sketching lower} can be improved to $\alpha\ge \exp\big(c(1-2\e)\sqrt{\log \log n}\big)$ for some universal constant $c>0$. Because it becomes a significantly more intricate case-specific  argument that does  not pertain to the more general geometric phenomenon that we study in Theorem~\ref{thm:nonembed main},  we will not include the technical details of this quantitative enhancement of Theorem~\ref{thm:our sketching statement} in the present extended abstract (the full version will contain more information).
\end{remark}

\subsection{Metric embeddings}\label{sec:embeddings} The distortion of a metric space $(U,d_U)$ in a metric space $(V,d_V)$ is a numerical invariant that is denoted $\cc_{(V,d_V)}(U,d_U)$ and defined to be the infimum over those $\alpha\in [1,\infty]$ for which there exist an embedding $f:U\to V$ and a scaling factor $\lambda\in (0,\infty)$ such that $\lambda d_U(x,y)\le d_V\big(f(x),f(y)\big)\le \alpha\lambda d_U(x,y)$ for all distinct $x,y\in U$. Given $p\ge 1$, the infimum of $\cc_{(V,d_V)}(U,d_U)$ over all possible\footnote{When $(U,d_U)$ is a finite metric space, it suffices to consider embeddings into $\ell_p$ rather than a general $L_p(\mu)$ space, as follows via a straightforward approximation by simple functions. We warn that this is not so for general (infinite) separable metric spaces, in which case one must consider embeddings into $L_p$; by~\cite[Corollary~1.5]{CK13} there is even a doubling subset of $L_1$ that does not admit a bi-Lipschitz embedding into $\ell_1$.} $L_p(\mu)$ spaces $(V,d_V)$ is denoted $\cc_p(U,d_U)$.

\subsubsection{Snowflakes} Because for every $\e\in (0,1)$ the quasi-norm $\|\cdot\|_{1-\e}$ does not induce a metric on $\ell_{1-\e}$, the embedding requirement~\eqref{eq:metric version into 1-eps} does not fit into the above standard metric embedding framework. However, as we explain in Fact~\ref{fact:snowflake} below, it is possible to situate~\eqref{eq:metric version into 1-eps}  within this framework (even without mentioning $\ell_{1-\e}$ at all) by considering embeddings of the $(1-\e)$-snowflake of a finite metric space into $\ell_1$. Recall the commonly used terminology (see e.g.~\cite{DS97}) that the $(1-\e)$-snowflake of a metric space $(M,d_M)$ is the metric space $(M,d_M^{1-\e})$.

\begin{fact}\label{fact:snowflake} Let $(M,d_M)$ be a finite\footnote{The only reason  for the finiteness assumption here (the present article deals only with finite metric space) is to ensure that the embedding is into $\ell_{1-\e}$ rather than a more general $L_{1-\e}(\mu)$ space. For embeddings of finite-dimensional normed spaces, i.e., the setting of~\cite{AKR18}, a similar reduction to embeddings into $\ell_{1-\e}$ is possible using tools from~\cite{Nik72,AMM85,BL00}.} metric space and fix $\e\in (0,1)$. The quantity $\cc_1(M,d_M^{1-\e})^{\frac{1}{1-\e}}$ is equal to the infimum over those $\alpha\ge 1$ for which there exists an embedding $f:M\to \ell_{1-\e}$ that satisfies~\eqref{eq:metric version into 1-eps}.
\end{fact}

\begin{proof} Suppose that $f:M\to \ell_{1-\e}$ satisfies~\eqref{eq:metric version into 1-eps}.  Then, recalling the notation~\eqref{eq:metric on <1} for the metric $\mathfrak{d}_{1-\e}$ on $\ell_{1-\e}$, we have $d_M(x,y)^{1-\e}\le \mathfrak{d}_{1-\e}(f(x),f(y))\le \alpha^{1-\e} d_M(x,y)^{1-\e}$ for all $x,y\in M$. It follows from general principles~\cite{BDK65,WW75} that the metric space $(\ell_{1-\e},\mathfrak{d}_{1-\e})$ admits an isometric embedding into an $L_1(\mu)$ space (an explicit formula for such an embedding into $L_1(\R^2)$ can be found in~\cite[Remark~5.10]{MN04}). Hence, $\cc_1(M,d_M^{1-\e})\le \alpha^{1-\e}$. Conversely, there is an explicit embedding (see equation~(2) in~\cite{Nao10}) $T:\ell_{1}\to L_{1-\e}(\N\times \R)$ which is an isometry when one takes the metric $\mathfrak{d}_{1-\e}$ on $L_{1-\e}(\N\times \R)$. Hence, if $\beta> \cc_1(M,d_M^{1-\e})$, then take an embedding $g:M\to \ell_1$ such that $d_M(x,y)^{1-\e}\le \|g(x)-g(y)\|_1\le \beta d_M(x,y)^{1-\e}$ for all $x,y\in M$ and consider the embedding $T\circ g$ which satisfies~\eqref{eq:metric version into 1-eps} with $\alpha=\beta^{1/(1-\e)}$, except that the target space is $L_{1-\e}(\N\times \R)$ rather than $\ell_{1-\e}$. By an approximation by simple functions we obtain the desired embedding into $\ell_{1-\e}$.
\end{proof}

\subsubsection{Quotients}\label{sec:quotients} Suppose that $G$ is a group that acts on a metric space $(X,d_X)$ by isometries. The quotient space $X/G=\{Gx\}_{x\in X}$ of all the orbits of $G$ can be equipped  with the following quotient metric $d_{X/G}:(X/G)\times (X/G)\to [0,\infty)$.
\begin{equation}\label{eq:def q metric}
\forall\, x,y\in X,\qquad d_{G/X}(Gx,Gy)\eqdef \inf_{(u,v)\in (Gx)\times (Gy)} d_X(u,v)=\inf_{g\in G} d_X(gx,y).
\end{equation}
See~\cite[Section~5.19]{BH99} for more on this basic construction (in particular, for a verification that~\eqref{eq:def q metric} indeed gives a metric).

Given $k\in \N$, we will consider the Hamming cube to be the vector space $\F_2^k$ over the field of two elements $\F_2$, equipped with the Hamming metric $d_{\F_2^k}:\F_2^k\times \F_2^k\to \N\cup \{0\}$ that is given by
$$
\forall\, x=(x_1,\ldots,x_k),y=(y_1,\ldots,y_k)\in \F_2^k,\qquad d_{\F_2^k}(x,y)=|\{j\in \k:\ x_j\neq y_j\}|.
$$
 Below, $\F_2^k$  will always be assumed to be equipped with the metric $d_{\F_2^k}$. The standard basis of $\F_2^k$ is denoted $e_1,\ldots,e_k$.

If $G$ is a group acting on $\F_2^k$ by isometries, and if it isn't too large, say, $|G|\le 2^{k/2}$, then all but an exponentially small fraction of the pairs $(x,y)\in \F_2^k\times \F_2^k$ satisfy $d_{\F_2^k}(Gx,Gy)\gtrsim k$. Specifically, there is a universal constant $\eta>0$ such that
\begin{equation}\label{eq:large typical distance}
|G|\le 2^{\frac{k}{2}}\implies \Big|\Big\{(x,y)\in \F_2^k\times \F_2^k:\ d_{\F_2^k/G}(x,y)\le \eta k\Big\}\Big|\le 2^{\frac53  k}.
\end{equation}
A simple counting argument which verifies~\eqref{eq:large typical distance} appears in the proof of~\cite[Lemma~3.2]{KN06}.

 The symmetric group $S_k$ acts isometrically on  $\F_2^k$ by permuting the coordinates, namely for each permutation $g$ of $\k$ and $x\in \F_2^k$ we write $gx=(x_{g^{-1}(1)},x_{g^{-1}(2)},\ldots,x_{g^{-1}(k)})$. A subgroup $G\le S_k$ of $S_k$ therefore acts by isometries on $\F_2^k$; below we will only consider quotients of the form $(\F_2^k/G,d_{\F_2^K/G})$ when $G$ is a transitive subgroup of $S_k$.

 \subsubsection{$\ell_1$ non-embeddability of snowflakes of (subsets of) hypercube quotients} In~\cite{KN06} we studied the $\ell_1$ embeddability of  quotients of $\F_2^k$. In particular, \cite[Corollary~3]{KN06} states that if $G$ is a transitive subgroup of $S_k$ with $|G|\le 2^{k/2}$, then
\begin{equation}\label{eq:prev quotient}
 \cc_1\big(\F_2^k/G,d_{\F_2^k/G}\big)\gtrsim \log k.
\end{equation}

In Remark~4 of~\cite{KN06} we (implicitly) asked about the sketchability of $\F_2^k/G$, by inquiring whether  its $(1/2)$-snowflake embeds into a Hilbert space with $O(1)$ distortion, as a possible alternative approach for obtaining integrality gaps (quantitatively stronger than what was known at the time) for the Goemans--Linial semidefinite relaxation of the Sparsest Cut problem. This hope was  realized in~\cite{DKSV06} for the special case when $G=\langle \mathfrak{S}_k\rangle\le S_k$ is the cyclic group that is generated by the cyclic shift $\mathfrak{S}_k=(1,2,\ldots,k)\in S_k$. Specifically, it follows from~\cite{DKSV06}  that there exists a large subset $M\subset \F_2^k$, namely $|\F_2^k\setminus M|\lesssim 2^k/k^2$, and a metric $\rho$ on $M/\langle \mathfrak{S}_k\rangle$ satisfying $\rho(\O,\O')\asymp d_{\F_2^k/\langle \mathfrak{S}_k\rangle}(\O,\O')$ for all pairs of orbits $\O,\O'\in M/\langle \mathfrak{S}_k\rangle$, and such that the metric space $(M/\langle \mathfrak{S}_k\rangle,\sqrt{\rho})$ embeds isometrically into $\ell_2$. Strictly speaking, a stronger statement than this was obtained in~\cite{DKSV06} for a larger metric space (namely, for the quotient of $\F_2^k\times \F_2^k$ by the group $\langle \mathfrak{S}_k\rangle\times \langle \mathfrak{S}_k\rangle$), but here it suffices to consider the above smaller metric space which inherits the stated properties.

Recalling Fact~\ref{fact:snowflake}, this discussion leads naturally, as a strategy towards proving Theorem~\ref{thm:our sketching statement}, to investigating whether a  lower bound as~\eqref{eq:prev quotient} holds for the $(1-\e)$-snowflake of the hypercube quotient $\F_2^k/G$ rather than that quotient itself. We will see that the method of~\cite{KN06} does not yield any such lower bound that tends to $\infty$ as $k\to \infty$ for fixed $\e>0$, but we do obtain the desired statement here, albeit with an asymptotically weaker lower bound than the $\log k$ of~\eqref{eq:prev quotient}. Note that an application of Theorem~\ref{thm:nonembed main} below to the above subset $M\subset \F_2^k$ from~\cite{DKSV06} yields Theorem~\ref{thm:our sketching statement}, because of Fact~\ref{fact:snowflake}.

\begin{theorem}[non-embeddability  of snowflakes of quotients of large subsets of the hypercube]\label{thm:nonembed main} Fix $k\in \N$ and $\e\in (0,\frac12)$. Let  $G$ be a transitive subgroup of $S_k$  with $|G|\le 2^{k/2}$. Then, every $M\subset \F_2^k$ with $|\F_2^k\setminus M|\le 2^k/\sqrt{\log k}$ satisfies
\begin{equation}\label{eq:lower snowflake}
\cc_1\Big(M/G,d_{\F_2^k/G}^{1-\e}\Big)\gtrsim (\log k)^{\frac12-
\e}.
\end{equation}
\end{theorem}

It would be interesting to determine the asymptotically sharp behavior (up to universal constant factors) in~\eqref{eq:lower snowflake} for $M=\F_2^k$, though understanding the dependence on the transitive subgroup $G\le S_k$ may be  challenging; see~\cite{BK97} for investigations along these lines. Even in the special case $G=\langle \mathfrak{S}_k\rangle$ we do not know the sharp bound, and in particular how it transitions from the  $(\log k)^{1/2-\e}$ of~\eqref{eq:lower snowflake}  to the $\log k$ of~\eqref{eq:prev quotient} as $\e\to 0$ (it could be that neither bound is tight).

\subsubsection{Bourgain's Fourier tails versus the Kahn--Kalai--Linial influence of variables}\label{Bourgain} In~\cite[Theorem~3.8]{KN06} we applied the important theorem~\cite{KKL88} of Kahn, Kalai and Linial on the influence of variables on Boolean functions to show that if $G$ is a transitive subgroup of $S_k$, then every $f:\F_2^k/G\to \ell_1$ satisfies the following Cheeger/Poincar\'e inequality.
\begin{equation}\label{eq:quote KN KKL}
\frac{1}{4^k}\sum_{(x,y)\in \F_2^k\times \F_2^k} \big\|f(Gx)-f(Gy)\big\|_1\lesssim \frac{1}{\log k}\sum_{j=1}^k\frac{1}{2^k}\sum_{x\in \F_2^k}\left\|f\big(G(x+e_j)\big)-f(Gx)\right\|_1.
\end{equation}
Fix $(\e,\alpha)\in (0,1)\times [1,\infty)$. If $|G|\le 2^{k/2}$ and   $d_{\F_2^k/G}(Gx,Gy)^{1-\e}\le \|f(Gx)-f(Gy)\|_1\le \alpha d_{\F_2^k/G}(Gx,Gy)^{1-\e}$ for $x,y\in \F_2^k$, then
\begin{multline*}
k^{1-\e}\stackrel{\eqref{eq:large typical distance}}{\lesssim} \frac{1}{4^k}\sum_{(x,y)\in \F_2^k\times \F_2^k} d_{\F_2^k/G}(Gx,Gy)^{1-\e}\le \frac{1}{4^k}\sum_{(x,y)\in \F_2^k\times \F_2^k} \big\|f(Gx)-f(Gy)\big\|_1\stackrel{\eqref{eq:quote KN KKL}}{\lesssim}\frac{1}{\log k}\sum_{j=1}^k\frac{1}{2^k}\sum_{x\in \F_2^k}\left\|f\big(G(x+e_j)\big)-f(Gx)\right\|_1\\\le  \frac{\alpha}{\log k}\sum_{j=1}^k\frac{1}{2^k}\sum_{x\in \F_2^k}d_{\F_2^k/G}\big(G(x+e_j),Gy\big)^{1-\e}\stackrel{\eqref{eq:def q metric}}{\le} \frac{\alpha}{\log k}\sum_{j=1}^k\frac{1}{2^k}\sum_{x\in \F_2^k}d_{\F_2^k}(x+e_j,y)^{1-\e}=\frac{\alpha k}{\log k}.
\end{multline*}
It follows that
\begin{equation}\label{eq:lower snowflake KKL}
\cc_1\Big(\F_2^k/G,d_{\F_2^k/G}^{1-\e}\Big)\gtrsim \frac{\log k}{k^\e}.
\end{equation}
This is how~\eqref{eq:prev quotient} was derived in~\cite{KN06}, but the right hand side of~\eqref{eq:lower snowflake KKL} tends to $\infty$ as $k\to \infty$ only if $\e\lesssim (\log \log k)/\log k$.

Following the above use of the KKL theorem~\cite{KN06}, it was used elsewhere in place of applications~\cite{KV15,KN06} of a more substantial theorem of Bourgain~\cite{Bou02} on the Fourier tails of Boolean functions that are not close to juntas; notably this was first done by Krauthgamer and Rabani~\cite{KR09} to obtain an asymptotically improved analysis of the Khot--Vishnoi integrality gap~\cite{KV15} for Sparsest Cut. We have seen above that the KKL-based  approach does not yield Theorem~\ref{thm:nonembed main} (though, of course, one cannot rule out the availability of a more sophisticated application of KKL that does), but our use of Bourgain's theorem in the ensuing proof of Theorem~\ref{thm:nonembed main}  shows that this theorem does sometime provide qualitatively stronger geometric information. One should note here that~\eqref{eq:lower snowflake} follows from an application of a sharp form of Bourgain's theorem that was more recently obtained by Kindler, Kirshner, and O'Donnell~\cite{KKO18}; an application of Bourgain's original formulation yields a bound that is asymptotically weaker by a lower-order factor.

\section{Proof of Theorem~\ref{thm:nonembed main}}

Here we will prove Theorem~\ref{thm:nonembed main}, thereby completing the justification of Theorem~\ref{thm:our sketching statement}  as well.

\subsection{Fourier-analytic preliminaries} We will include here some basic facts and notation related to Fourier analysis on the hypercube $\F_2^k$; an extensive treatment of this topic can be found in e.g.~the monograph~\cite{Odo14}. Fix $k\in \N$. From now on, let $\mu=\mu_k$ denote the normalized counting measure on $\F_2^k$. Given $A\subset \k$, the  Walsh function $\W_{\!\! A}:\F_2^k\to \{-1,1\}$ and Fourier coefficient $\widehat{\f}(A)\in \R$ of a function $\f:\F_2^k\to \R$ are defined by
$$
\forall\, x\in \F_2^k,\qquad \W_{\!\!A}(x)= (-1)^{\sum_{j=1}^n x_j}\qquad\mathrm{and}\qquad \widehat{\f}(A)= \int_{\F_2^k} \f(x)\W_{\!\!A}(x)\ud\mu(x).
$$
The convolution $\f*\psi:\F_2^k\to \R$  of two functions $\f,\psi:\F_2^k\to \R$ is defined by
$$
\forall\, x\in \F_2^k,\qquad (\f*\psi)(x)=\int_{\F_2^k} \f(y)\psi(x+y)\ud\mu(y)=\sum_{A\subset \k} \widehat{\f}(A)\widehat{\psi}(A)\W_{\!\!A}(x),
$$
where the last equality is valid because the $2^k$ Walsh functions $\{\W_{\!\!A}\}_{A\subset \k}$ consist of all of the characters of the additive group $\F_2^k$, hence forming an orthonormal basis of $L_2(\mu)$. Suppose that $g\in \mathsf{GL}(\F_2^k)$ is an automorphism of $\F_2^k$. If $\f:\F_2^k\to \R$ is a $g$-invariant function, i.e., $\f(gy)=\f(y)$ for all $y\in \F_2^k$, then for every $\psi:\F_2^k\to \R$ and $x\in \F_2^k$,
\begin{multline*}
(\f*\psi)(x)=\int_{\F_2^k}\f(y)\psi(x+y)\ud \mu (y)=\int_{\F_2^k}\f(gy)\psi(x+y)\ud \mu (y)\\=\int_{\F_2^k}\f(z)\psi\big(x+g^{-1}z\big)\ud \mu (z)=
\int_{\F_2^k}\f(z)\psi\big(g^{-1}(gx+z)\big)\ud \mu (z)=\Big(\f*\big(\psi\circ g^{-1}\big)\Big)(gx).
\end{multline*}
In particular, under the above invariance assumption we have the identity
\begin{equation}\label{eq:invariance under g}
\|\f*\psi\|_{L_2(\mu)}= \Big\|\f*\big(\psi\circ g^{-1}\big)\Big\|_{L_2(\mu)}.
\end{equation}

Given $\p\in [0,1]$, let $\vartheta^\p: 2^{\F_2^k\times \F_2^k}\to [0,1]$ be the probability measure that is defined by setting for each $(x,y)\in \F_2^k\times \F_2^k$,
\begin{equation}\label{eq:theta measure}
\vartheta^\p(x,y)\eqdef \frac{\p^{d_{\F_2^k}(x,y)}(1-\p)^{k-d_{\F_2^k}(x,y)}}{2^k}=\frac{1}{4^k}\prod_{j=1}^k \big(1+(1-2\p)(-1)^{x_j+y_j}\big)=\frac{1}{4^k}\sum_{A\subset \k} (1-2\p)^{|A|}\W_{\!\!A}(x+y).
\end{equation}
In other words, $\vartheta^\p(x,y)$ is equal to the probability that the ordered pair $(x,y)$ is the outcome of the following randomized selection procedure:  The first element $x\in \F_2^k$ is chosen uniformly at random, and the second element $y\in \F_2^k$ is obtained by changing the sign of each entry of $x$ independently with probability $\p$. Note in passing that  both marginals of $\t^\p$ are equal to $\mu$, i.e., $\t^\p(\Omega\times \F_2^k)=\t^\p(\F_2^k\times \Omega)=\mu(\Omega)$ for every $\Omega\subset \F_2^k$. Also, for every $\Omega\subset \F_2^k$ we have
\begin{equation}\label{eq:heat identity}
\begin{split}
\t^\p\big(\Omega\times (\F_2^k\setminus \Omega)\big)=\frac18 &\int_{\F_2^k\times \F_2^k} \Big((-1)^{\1_{\Omega}(x)}-(-1)^{\1_{\Omega}(y)}\Big)^2\ud \t^\p(x,y)\\&=\frac14 \Bigg(1-\int_{\F_2^k\times \F_2^k}(-1)^{\1_{\Omega}}(x)(-1)^{\1_{\Omega}}(y)\ud \t^\p(x,y)\Bigg)
=\frac14\sum_{A\subset\k} \Big(1- (1-2\p)^{|A|}\Big) \Big(\widehat{(-1)^{\1_{\Omega}}}(A)\Big)^2,
\end{split}
\end{equation}
where the last equality in~\eqref{eq:heat identity} is a direct consequence of Parseval's identity and the final expression in~\eqref{eq:theta measure} for $\vartheta^\p(\cdot,\cdot)$.

For $\f:\F_2^k\to \R$ and $j,m\in \k$, the level-$m$ influence of the $j$'th variable on $\f$, denoted $\mathsf{Inf}_j^{\le m}[\f]$, is the quantity
\begin{equation}\label{eq:jth influence}
\mathsf{Inf}_j^{\le m}[\f]= \sum_{\substack {A\subset \k\setminus\{j\}\\ |A|\le m-1}} \widehat{\f}(A\cup\{j\})^2=\Big\|\f*\mathscr{R}_j^{\le m}\Big\|_{L_2(\mu)}^2,
\end{equation}
where the last equality is a consequence of Parseval's identity, using the notation
\begin{equation}\label{eq:def Rjm}
\mathscr{R}_j^{\le m}\eqdef \sum_{\substack {A\subset \k\setminus\{j\}\\ |A|\le m-1}} \W_{\!\! A\cup\{j\}}.
\end{equation}
It follows from the first equation in~\eqref{eq:jth influence} that
\begin{equation}\label{eq:m variance}
\sum_{j=1}^k \mathsf{Inf}_j^{\le m}[\f]= \sum_{\substack {B\subset \k\\ |B|\le m}} |B|\widehat{\f}(B)^2\le m\sum_{\substack{B\subset \k\\ B\neq\emptyset}} \widehat{\f}(B)^2=m \Bigg(\int_{\F_2^k} \f^2\ud\mu-\widehat{\f}(\emptyset)^2\Bigg)=m\mathsf{Var}_\mu[\f],
\end{equation}
where $\mathsf{Var}_\mu[\cdot]$ denotes the variance with respect to the probability measure $\mu$. By considering the symmetric group $S_k$ as a subgroup of $\mathsf{GL}(\F_2^k)$, where the action is permutation of coordinates, an inspection of  definition~\eqref{eq:def Rjm} reveals that $\mathscr{R}_j^{\le m}\circ g= \mathscr{R}_{gj}^{\le m}$ for  $g\in S_k$ and $j,m\in \k$. By~\eqref{eq:invariance under g} and the second equality in~\eqref{eq:jth influence}, if $\f:\F_2^k\to \R$ is $g$-invariant, then
$$
\forall\, j,m\in\k,\qquad \mathsf{Inf}_j^{\le m}[\f]=\mathsf{Inf}_{g^{-1}j}^{\le m}[\f].
$$
A combination of this observation with~\eqref{eq:m variance} yields the following statement, which we record for ease of later reference.

\begin{fact}\label{fact} Fix $k\in \N$. Let $G$ be a subgroup of $S_k$ that acts transitively on the coordinates $\k$. Suppose that $\f:\F_2^k\to \R$ is a $G$-invariant function, i.e., $f(gx)=f(x)$ for every $g\in G$ and $x\in \F_2^n$. Then, for every $m\in \k$ we have
$$
\max_{j\in \k} \mathsf{Inf}_j^{\le m}[\f]\le \frac{m}{k}\mathsf{Var}_\mu[\f].
$$
\end{fact}

Throughout what follows, given a subgroup $G\le S_k$, we denote by $\pi_G:\F_2^k\to \F_2^k/G$ its associated quotient mapping, i.e., $\pi_G(x)=Gx$ for all $x\in \F_2^k$. We denote by $\mu_{\F_2^k/G}=\mu\circ\pi_G^{-1}$ the probability measure on $\F_2^k/G$ that is given by
$$
\forall\, \O\in \F_2^k/G,\qquad \mu_{\F_2^k/G}(\O)=\mu(\O).
$$
In a similar vein, for every $\p\in [0,1]$ the probability measure $\t^\p$ on $\F_2^k\times \F_2^k$ that is given in~\eqref{eq:theta measure} descends to a probability measure $\t^\p_{\F_2^k/G}=\t^\p\circ (\pi_G\times \pi_G)^{-1}$ on $(\F_2^k/G)\times (\F_2^k/G)$ by setting
$$
\forall\, \O,\O'\subset \F_2^k/G,\qquad \t^\p_{\F_2^k/G}(\O,\O')=\t^\p(\O\times \O').
$$

\subsection{A Cheeger/Poincar\'e inequality for transitive  quotients} Our main technical result is the following inequality.

\begin{lemma}\label{lem:poincare on Y}There is a universal constant $\beta\in (0,1)$ with the following property. Fix an integer  $k\ge 55$ and a transitive subgroup $G$ of $S_k$. Suppose that $X\subset \F_2^k/G$ is a sufficiently large subset in the following sense.
\begin{equation}\label{eq:mu lower assyumption sqrt}
\mu_{\F^k_2/G}(X)\ge 1-\frac{1}{\sqrt{\log k}}.
\end{equation}
Then there is a further subset $Y\subset X$ with $\mu_{\F_2^k/G}(Y)\ge \frac34\mu_{\F_2^k/G}(X)$ such that every  function $f:Y\to \ell_1$ satisfies
\begin{equation}\label{eq:poincare on Y}
\iint_{Y\times Y} \|f(\O)-f(\O')\|_1\ud \mu_{\F_2^k/G}(\O)\ud\mu_{\F_2^k/G}(\O')\lesssim \sqrt{\log k} \iint_{Y\times Y} \|f(\O)-f(\O')\|_1\ud\vartheta_{\F_2^n/G}^{\frac{1}{\beta\log k}}(\O,\O').
\end{equation}
\end{lemma}

Prior to proving Lemma~\ref{lem:poincare on Y} we shall assume its validity for the moment and proceed to prove Theorem~\ref{thm:nonembed main}.

\begin{proof}[Proof of Theorem~\ref{thm:nonembed main} assuming Lemma~\ref{lem:poincare on Y}] Fix $\alpha\ge 1$ and suppose that $f:M/G\to \ell_1$ satisfies
\begin{equation}\label{eq:f assumptions in final conclusion}
\forall\,  x,y\in M,\qquad d_{\F_2^k/G}(Gx,Gy)^{1-\e}\le \|f(Gx)-f(Gy)\|_1\le \alpha d_{\F_2^k/G}(Gx,Gy)^{1-\e}.
\end{equation}
Our task is to bound $\alpha$ from below by the right hand side of~\eqref{eq:lower snowflake}.

An application of  Lemma~\ref{lem:poincare on Y} to $X=M/\F_2$, which satisfies the requirement~\eqref{eq:mu lower assyumption sqrt} by the assumption of Theorem~\ref{thm:nonembed main}, produces a subset $Y$ with $\mu(\pi^{-1}_G(Y))\ge \frac12$ for which~\eqref{eq:poincare on Y} holds true. It follows that
\begin{equation}\label{eq:use poincare on Y}
\begin{split}
&\!\!\!\iint_{\pi_G^{-1}(Y)\times \pi_G^{-1}(Y)}d_{\F_2^k/G}(Gx,Gy)^{1-\e}\ud\mu(x)\ud\mu(y)\stackrel{\eqref{eq:poincare on Y}\wedge\eqref{eq:f assumptions in final conclusion}} {\lesssim} \alpha\sqrt{\log k}\iint_{(\F_2^k/G)\times (\F_2^k/G)} d_{\F_2^k/G}(\O,\O')^{1-\e}\ud\vartheta_{\F_2^n/G}^{\frac{1}{\beta\log k}}(\O,\O')\\&\stackrel{\eqref{eq:def q metric}}{\le} \alpha\sqrt{\log k}\int_{\F_2^k\times \F_2^k} d_{\F_2^k} (x,y)^{1-\e}\ud\vartheta^{\frac{1}{\beta\log k}}(x,y)\stackrel{\eqref{eq:theta measure}}{=}\alpha\sqrt{\log k}\sum_{\ell=0}^k \ell^{1-\e}\binom{k}{\ell} \Bigg(\frac{\beta}{\log k}\Bigg)^\ell \Bigg(1-\frac{\beta}{\log k}\Bigg)^{k-\ell}\le \alpha\sqrt{\log k}\Bigg(\frac{\beta k}{\log k}\Bigg)^{1-\e}.
\end{split}
\end{equation}
Since  $|G|\le 2^{k/2}$, by~\eqref{eq:large typical distance} there exists $\eta\gtrsim 1$ such that, since $\mu(\pi_G^{-1}(Y))\gtrsim 1$, we have
\begin{multline*}
\mu\times \mu \Big(\big\{(x,y)\in \pi_G^{-1}(Y)\times \pi_G^{-1}(Y):\ d_{F_2^k}(Gx,Gy)> \eta k\big\}\Big)\\\ge \mu\big(\pi_G^{-1}(Y)\big)^2- \mu\times \mu \Big(\big\{(x,y)\in \F_2^k\times \F_2^k:\ d_{F_2^k}(Gx,Gy)\le \eta k\big\}\Big)\ge \mu\big(\pi_G^{-1}(Y)\big)^2-2^{-\frac{k}{3}}\gtrsim 1.
\end{multline*}
So, the first quantity in~\eqref{eq:use poincare on Y} is at least a constant multiple of $k^{1-\e}$, and the desired lower bound on $\alpha$ follows.
\end{proof}

\begin{proof}[Proof of Lemma~\ref{lem:poincare on Y}] Suppose that $Z\subset X$ satisfies
\begin{equation}\label{eq:ratio Z}
\frac{1}{4}\le \frac{\mu_{\F^k/G}(Z)}{\mu_{\F^k/G}(X)}\le \frac23.
\end{equation}
Writing $\q=\mu_{\F^k/G}(Z)$,  the function $(-1)^{\1_{\pi_G^{-1}(Z)}}:\F_2^k\to \{-1,1\}$ is $G$-invariant and its variance is equal to $4\q(1-\q)\asymp 1$. Let $\beta\in (2/\log k,1)$ be a small enough universal constant that will be determined later. Also, let $C\in (1,\infty)$ be a large enough universal constant, specifically take $C$ to  be the universal constant that appears in the statement of~\cite[Theorem~3.1]{KKO18}. If we denote $m=\lceil \beta \log k\rceil$, then it follows from Fact~\ref{fact} that, provided $\beta$ is a sufficiently small constant, we have
$$
\max_{j\in \k} \mathsf{Inf}_j^{\le m}\Big[(-1)^{\1_{\pi_G^{-1}(Z)}}\Big]\le \frac{m}{k}\mathsf{Var}\Big[(-1)^{\1_{\pi_G^{-1}(Z)}}\Big]\le \frac{\mathsf{Var}\Big[(-1)^{\1_{\pi_G^{-1}(Z)}}\Big]^4}{C^m}.
$$
This is precisely the assumption of~\cite[Theorem~3.1]{KKO18}, from which we deduce  the following Fourier tail bound.
\begin{equation}\label{eq:use Bourgain}
\sum_{\substack{A\subset \k\\ |A|>\lceil \beta \log k\rceil}} \Bigg(\widehat{(-1)^{\1_{\pi_G^{-1}(Z)}}}(A)\Bigg)^2=\sum_{\substack{A\subset \k\\ |A|>m}} \Bigg(\widehat{(-1)^{\1_{\pi_G^{-1}(Z)}}}(A)\Bigg)^2\gtrsim \frac{\mathsf{Var}\Big[(-1)^{\1_{\pi_G^{-1}(Z)}}\Big]}{\sqrt{m}}\asymp \frac{1}{\sqrt{\beta\log k}}.
\end{equation}

Next, by the identity~\eqref{eq:heat identity}, we have
\begin{equation}\label{eq:edges going out but maybe not in X}
\begin{split}
\t^{\frac{1}{\beta\log k}}\Big(\pi_G^{-1}(Z)\times \big(\F_2^k\setminus \pi_G^{-1}(Z)\big)\Big)=\frac14&\sum_{A\subset\k} \Bigg(1- \Bigg(1-\frac{2}{\beta\log k}\Big)^{|A|} \Bigg)\Bigg(\widehat{(-1)^{\1_{\pi_G^{-1}(Z)}}}(A)\Bigg)^2\\&\ge \frac14\Bigg(1- \Big(1-\frac{2}{\beta\log k}\Big)^{\lceil \beta \log k\rceil+1} \Bigg)\sum_{\substack{A\subset \k\\ |A|> \lceil \beta \log k\rceil}}\Bigg(\widehat{(-1)^{\1_{\pi_G^{-1}(Z)}}}(A)\Bigg)^2\stackrel{\eqref{eq:use Bourgain}}{\ge} \frac{\gamma}{\sqrt{\beta \log k}},
\end{split}
\end{equation}
for some universal constant $\gamma\in (0,1)$. Therefore,
\begin{equation}\label{eq:correct for edges outside}
\begin{split}
\t^{\frac{1}{\beta\log k}}\Big(\pi_G^{-1}&(Z)\times \big(\pi^{-1}_G(X)\setminus \pi_G^{-1}(Z)\big)\Big)\ge \t^{\frac{1}{\beta\log k}}\Big(\pi_G^{-1}(Z)\times \big(\F_2^k\setminus \pi_G^{-1}(Z)\big)\Big)- \t^{\frac{1}{\beta\log k}}\Big(\F_2^k\times \big((\F_2^k\setminus \pi_G^{-1}(X)\big)\Big)\\&= \t^{\frac{1}{\beta\log k}}\Big(\pi_G^{-1}(Z)\times \big(\F_2^k\setminus \pi_G^{-1}(Z)\big)\Big)-\mu\big(\F_2^k\setminus \pi_G^{-1}(X)\big)\stackrel{\eqref{eq:mu lower assyumption sqrt}\wedge \eqref{eq:edges going out but maybe not in X}}{\ge} \frac{\gamma}{\sqrt{\beta \log k}}-\frac{1}{\sqrt{\log k}}\stackrel{\eqref{eq:ratio Z}}{\asymp} \frac{1}{\sqrt{\log k}}\cdot \frac{\mu_{\F^k/G}(Z)}{\mu_{\F^k/G}(X)},
\end{split}
\end{equation}
where the final step of~\eqref{eq:correct for edges outside} holds provided $1\asymp \beta\le \gamma^2/4$, which is our final requirement from the universal constant $\beta$.

Observe that
\begin{equation}\label{eq:theta has big mathh on X times X}
\begin{split}
\t_{\F_2^k/G}^{\frac{1}{\beta\log k}}(X\times X)\ge \t^{\frac{1}{\beta\log k}}\big(\F_2^k\times \F_2^k\big)-\t^{\frac{1}{\beta\log k}}\Big(\big(\F_2^k\setminus \pi_G^{-1}&(X)\big)\times \F_2^k\Big)-\t^{\frac{1}{\beta\log k}}\Big(\F_2^k\times \big(\F_2^k\setminus \pi_G^{-1}(X)\big)\Big)\\ &=1-2\mu\big(\F_2^k\setminus \pi_G^{-1}(X)\big)=1-2\big(1-\mu_{\F_2^k/G}(X)\big)\stackrel{\eqref{eq:mu lower assyumption sqrt}}{\ge} 1-\frac{2}{\sqrt{\log k}}\asymp 1.
\end{split}
\end{equation}
Hence,
\begin{equation}\label{eq:big subset expand}
\frac{\t_{\F_2^k/G}^{\frac{1}{\beta\log k}}\Big(\big(Z\times (X\setminus Z)\big)\cup \big((X\setminus Z)\times Z\big) \Big) }{\t_{\F_2^k/G}^{\frac{1}{\beta\log k}}(X\times X)}=\frac{2\t^{\frac{1}{\beta\log k}}\Big(\pi_G^{-1}(Z)\times \big(\pi^{-1}_G(X)\setminus \pi_G^{-1}(Z)\big)\Big) }{\t_{\F_2^k/G}^{\frac{1}{\beta\log k}}(X\times X)}\stackrel{\eqref{eq:correct for edges outside}}{\gtrsim} \frac{1}{\sqrt{\log k}}\cdot \frac{\mu_{\F^k/G}(Z)}{\mu_{\F^k/G}(X)}.
\end{equation}
We are now in position to apply~\cite[Lemma~6]{KN06} with the parameters  $\d=\frac14$, $\alpha\asymp1/\sqrt{\log k}$, and  the probability measures
\begin{equation}\label{eq:def sigma tau}
\sigma\eqdef \frac{\mu_{\F_2^k/G}}{\mu_{\F_2^k/G}(X)}: 2^X\to [0,1]\qquad\mathrm{and}\qquad \tau\eqdef \frac{\t_{\F_2^k/G}^{\frac{1}{\beta\log k}}}{\t_{\F_2^k/G}^{\frac{1}{\beta\log k}}(X\times X)}: 2^{X\times X}\to [0,1].
\end{equation}
Due to~\eqref{eq:big subset expand}, by the proof of~\cite[Lemma~6]{KN06} (specifically, equation~(7) in~\cite{KN06}) there exists a subset $Y\subset \F_2^k/G$ with $\sigma(Y)\ge 3/4$, i.e., $\mu_{\F_2^k/G}(Y)\ge 3\mu_{\F_2^k/G}(X)/4$, such that every $f:Y\to L_1$ satisfies
\begin{multline*}
\iint_{Y\times Y} \|f(\O)-f(\O')\|_1\ud \mu_{\F_2^k/G}(\O)\ud\mu_{\F_2^k/G}(\O') \stackrel{\eqref{eq:ratio Z}\wedge \eqref{eq:def sigma tau}}{\asymp}\iint_{Y\times Y} \|f(\O)-f(\O')\|_1\ud \sigma(\O)\ud\sigma(\O')\\
\lesssim \sqrt{\log k}\iint_{Y\times Y} \|f(\O)-f(\O')\|_1\ud\tau(\O,\O') \stackrel{\eqref{eq:theta has big mathh on X times X}\wedge \eqref{eq:def sigma tau}}{\asymp} \sqrt{\log k} \iint_{Y\times Y} \|f(\O)-f(\O')\|_1\ud\vartheta_{\F_2^n/G}^{\frac{1}{\beta\log k}}(\O,\O').\tag*{\qedhere}
\end{multline*}
\end{proof}


\bibliography{AKR}
\bibliographystyle{abbrv}

\end{document}